\documentclass[11pt]{article}
\usepackage{amssymb,amsthm,amsmath,tikz,pgf,mathtools}
\usepackage[lined,linesnumbered,ruled]{algorithm2e}
\usepackage{graphicx}\usepackage{amsfonts}
\usepackage{epstopdf}
\date{}

\usetikzlibrary{shapes.geometric}

\newtheorem{notat}{Notation}
\newtheorem{prop}{Proposition}
\newtheorem{defin}{Definition}

\title{Rooted Tree Arithmetic and Equations\thanks{Many thanks are due to Federico Poloni and Mahdi Amani for their comments and suggestions.}
}
\author {Fabrizio Luccio
\\Dipartimento di Informatica, University of Pisa.
\\luccio@di.unipi.it 
}

\begin{document}
\maketitle


\begin{abstract}

We propose a new arithmetic for non-empty rooted unordered trees simply called trees. 
After discussing tree representation and enumeration, 
we define the operations of tree addition, multiplication and stretch, prove their properties, and show that all trees can be generated from a starting tree of one vertex. 
We then show how a given tree can be obtained as the sum or product of two trees, thus defining {\em prime trees} with respect to addition and multiplication. In both cases we
show how primality can be decided in time polynomial in the number of vertices and we prove that factorization is unique.  
We then define negative trees and suggest dealing with tree equations, giving some preliminary results.
Finally we comment on how our arithmetic might be useful, and discuss preceding studies that have some relations with our. 
To the best of our knowledge our approach and results are completely new aside for a similar proposal deposited as an arXiv manuscript~\cite{L15}.
\end{abstract}

\section{Basic properties and notation } \label{basic}

\begin{itemize}

\item We refer to  {\bf rooted unordered trees} simply called trees. 
Our trees are non empty. {\bf 1} denotes the tree containing exactly one vertex, and is the basic element of our theory.

\item In a tree $T$, {\bf \em r}$\,(T)$ denotes the root of $T$; $x \in T$ denotes any of its vertices;
$n_T$ and $e_T$ respectively denote the numbers of vertices and leaves. A subtree is the tree composed of a vertex $x$ and all its descendants in $T$. The subtrees routed at the children of  $x$ are called subtrees of $x$. 
$s_T$ denotes the number of subtrees of {\bf \em r}$\,(T)$.

\item A tree $T$ can be represented as a binary sequences $S_T$ (the original reference for ordered trees is \cite{Z80}). In our scheme $T$ is traversed in left to right preorder inserting 1 in the sequence for each vertex encountered, and inserting 0 for each move backwards. Then $S_T$ is composed of $2n$ bits as shown in Figure 1, and
has the recursive structure 1 $S_1$ . . . $S_k$ 0, where the $S_i$ are the sequences representing the subtrees of {\bf \em r}$\,(T)$. The sequences for tree {\bf 1} is 10.
Note that all the prefixes of $S_T$ have more 1's than 0's except for the whole sequence that has as many 1's as 0's.

\end{itemize}

\begin{figure}
\label{fig1}
\begin{center}
\includegraphics[scale=0.5]{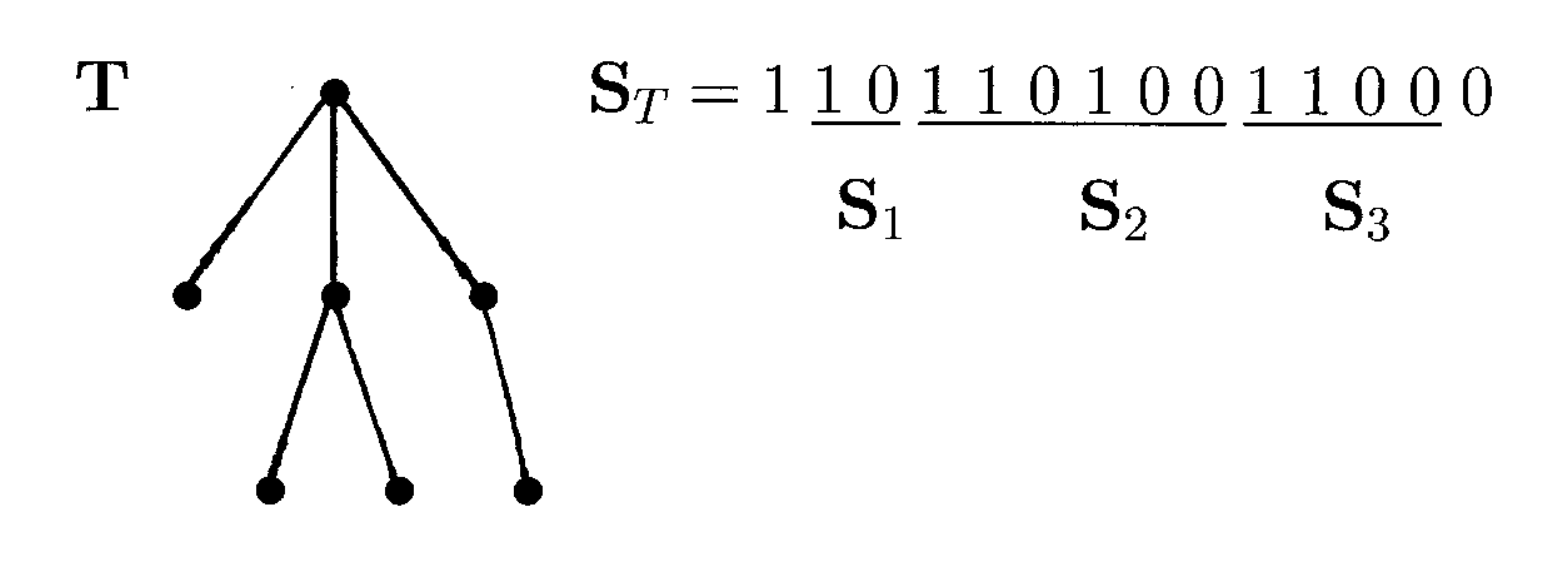}
\end{center}

\caption {Tree  representation as a binary sequence. $S_1, S_2, S_3$ represent the subtrees of the root of $T$.}
\end{figure}

Since $T$ is unordered, the order in which the subsequences $S_i$ appear in $S_T$ is immaterial (i.e., in general many different sequences represent $T$).
However a {\em canonical form} for trees is established so that their sequences will be uniquely determined, and will result  to be ordered
for increasing values if interpreted as binary numbers. To this end 
the trees are grouped into consecutive families ${\cal F}_1, {\cal F}_2,\dots$ as shown in Figure 2, where ${\cal F}_i$ contains the trees of $i$ vertices.
So the trees are ordered for increasing number of vertices, and inside each family  the ordering is determined by the canonical form as follows.
 Trees and sequences are then numbered with increasing natural numbers.

\begin{figure}
\label{fig4}
\begin{center}
\includegraphics[scale=0.40]{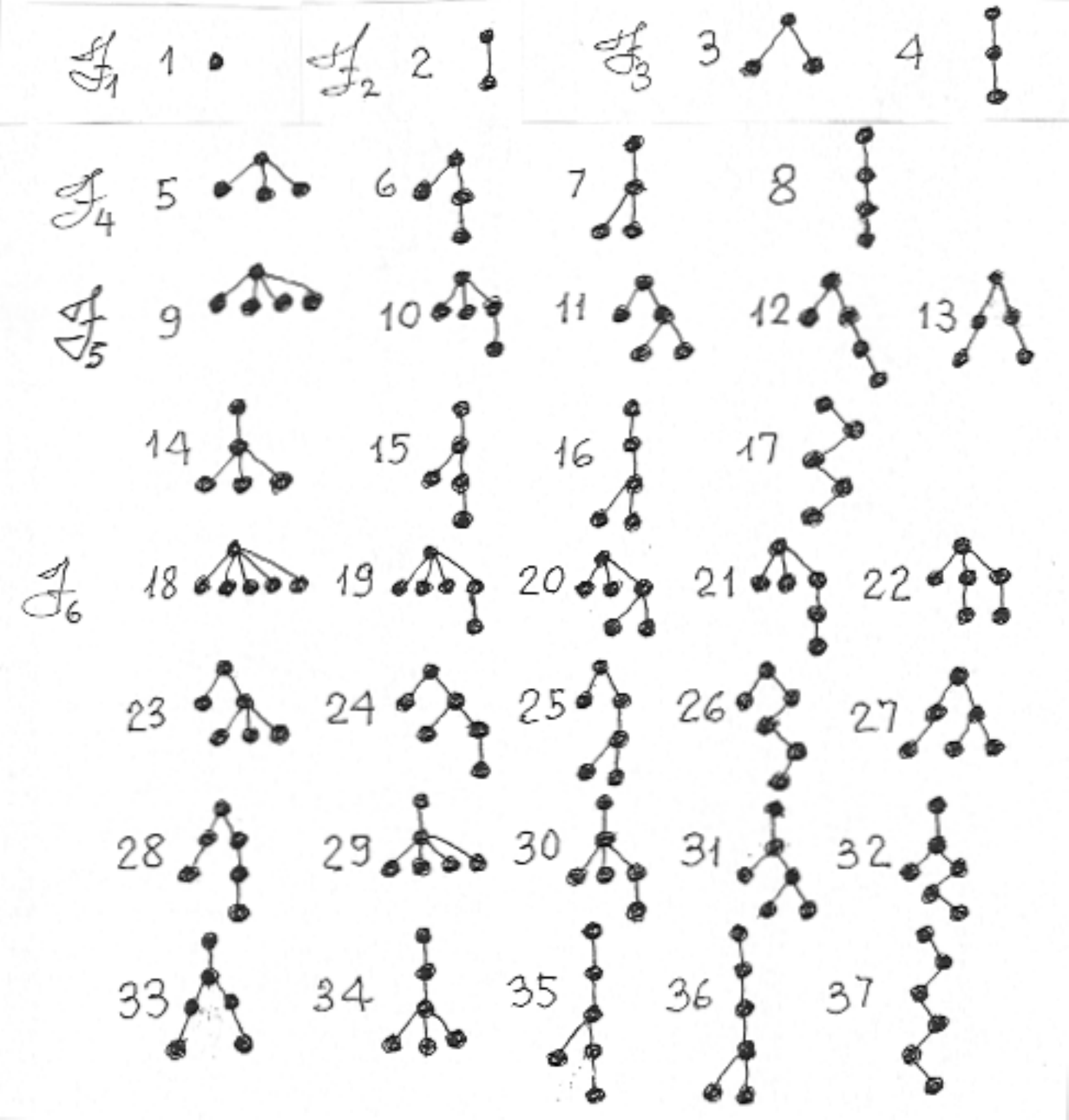}
\end{center} 

\caption {The canonical families of trees ${\cal F}_1$ to ${\cal F}_6$ and the corresponding tree enumeration.}
\end{figure}

\begin{itemize}

\item{If the sequences are interpreted as binary numbers, for two trees $U,T$ with $n_U<n_T$ we have $S_U<S_T$ because the initial character of each sequence is 1 and $S_U$ is shorter than $S_T$. This is consistent with the property that the trees of ${\cal F}_{n_U}$ precede the trees of ${\cal F}_{n_T}$ in the ordering.}

\item The families ${\cal F}_1,{\cal F}_2$ contain one tree each numbered 1, 2. 

\item The ordering of the trees in ${\cal F}_{n>2}$ is based on the ordering of the preceding families. Consider the multisets of positive integers whose sum is  $n-1$. E.g., for $n=6$ these multisets are: 1,1,1,1,1 - 1,1,1,2 - 1,1,3 - 1,2,2 - 1,4 - 2,3 - 5 ordered for non-decreasing value of the digits left to right. Each multiset corresponds to a group of consecutive trees in ${\cal F}_n$, where the digits in the multiset indicate the number of vertices of the subtrees of the root.
For ${\cal F}_6$ in Figure 2, multiset 1,1,1,1,1 refers to tree 18;  multiset 1,1,1,2 refers to tree 19; 
multiset 1,1,3 refers to trees 20 and 21 that have the two trees of ${\cal F}_3$ as third subtree, following the ordering in ${\cal F}_3$; $\dots$; multiset 2,3 refers to trees 27 and 28; the last multiset 5 refers to trees 29 to 37 whose roots have only one child.

\end{itemize}

So the first tree in ${\cal F}_n$ is the one of height 2 with $n-1$ subtrees of the root of one vertex each and sequence 1 1 0 1 0 1 0 . . . 1 0 0; and the last tree is the ``chain'' of $n$ vertices and sequence 1 1 . . . 1 0 0 . . . 0. As said the binary  sequences representing the trees in ${\cal F}_n$ are ordered for increasing values, see the listing for the first six canonical families in the Appendix. 

Many of these trees (not necessarily all) of each family ${\cal F}_n$ can be generated from the ones in ${\cal F}_{n-1}$ using the following:

\vspace{3mm}
\noindent{\bf Doubling Rule DR}. From each tree $T$ in ${\cal F}_{n-1}$ build two trees $T_1,T_2$ in ${\cal F}_n$ by adding a new vertex as the leftmost child of {\bf \em r}$\,(T)$, or adding a new root and appending $T$ to it as a unique subtree.

\vspace{3mm}
For example the four trees of ${\cal F}_4$ in Figure 2 can be built by {\bf DR} from the two trees of ${\cal F}_3$. The nine trees of ${\cal F}_5$ can be built by {\bf DR} from the four trees of ${\cal F}_4$, with the exception of tree 13. The twenty trees of ${\cal F}_6$ can be built by {\bf DR} from the nine trees of ${\cal F}_5$, with the exception of trees 27 and 28. In fact the number of extra trees that cannot be built with {\bf DR} increases sharply with $n$. 
Letting $f_n$ denote the number of trees in ${\cal F}_n$ 
we immediately have
$f_n\geq 2^{n-2}$ for $n\geq 2$. But a deep analysis~\cite{F03,PR94} has shown that the asymptotic value of this function is much higher, and can be approximated as: 

\vspace{3mm}
\hspace{5cm}$f_n  \sim 0.44 \cdot 2.96^n \cdot n^{-3/2}.$\hfill(1)

\vspace{3mm}
Then the minimum length of the sequences representing the trees of ${\cal F}_n$ is given approximately by:

$$\log_2(0.44 \cdot 2.96^n \cdot n^{-3/2}) \sim 1.57\,n-1.5\log_2n-1.19$$ 

\noindent much less than the $2n$ bits of our proposal. We only note that for $n\geq 2$ all the binary sequences representing our trees begin with two 1's and end with two 0's (see 
the listing in the Appendix), then these four digits could be removed, leaving a sequence of $2n-4$ bits to represent a tree. We shall see that our representation is amenable at working easily on the trees, so we maintain it, leaving the construction of a shorter efficient coding as a challenging open problem.

An arbitrary tree $T$ can be transformed into its canonical form with Algorithm CF of Figure 3. An elementary analysis shows that the algorithm is correct and each of its steps {\bf 1},$\,${\bf 2} can be executed in total O($n^2$) time. The algorithm can be possibly improved, however, our present aim is just showing that the problem can be solved in polynomial time.


\begin{figure} \label{CF}

\begin{center}
\fbox{
\begin{minipage}{14cm}
{\bf algorithm} CF($T,n$)

\vspace{1mm}


\vspace{1mm}
{\bf 1.} {\bf forany} vertex $x\in T$ 

\hspace{10mm}{\bf count} the number of vertices $n_1,\dots n_k$ of its subtrees;

\hspace{10mm}{\bf reorder} these subtrees for non decreasing values of the $n_i$;

\hspace{10mm}{\bf let} $G_1,\dots,G_r$ be the groups of subtrees with the same number  $g_1,\dots, g_r$

\hspace{17mm}of vertices, with all $g_i> 2$;

\vspace{1mm}
\hspace{5mm}// reordering is necessary but not sufficient for having $T$ in canonical form

\hspace{5mm}// the trees in all $G_i$ must be be arranged in canonical order

\vspace{1mm}
{\bf 2.} {\bf forany} $x\in T$, down-top from the vertices closest to the leaves

\hspace{10mm} {\bf forany} group $G_i=\{T_1,\dots,T_s\}$ 

\hspace{17mm} {\bf compute} the representing sequences $S_1,\dots,S_s$; 

\hspace{17mm} {\bf order} $S_1,\dots,S_s$ for increasing binary value; 

\hspace{17mm} {\bf permute} $T_1,\dots,T_s$ accordingly.

\end{minipage}
}

\end{center}

\caption {\small Structure of Algorithm CF for transforming an arbitrary tree $T$ of $n$ vertices in canonical form. CF requires polynomial time in $n$.}
\end{figure}


\section{Operators and tree generation} \label{operators}

Our basic operations are addition (symbol +) and multiplication (symbol \large$\cdot$\normalsize$\,$, or simple concatenation) defined as follows. Referring to Figure 4, let $A,B$ be two arbitrary trees:

\begin{itemize}

\item{\bf Addition}. $T=A+B$ is built by merging the two roots {\bf \em r}$\,(A)$, {\bf \em r}$\,(B)$ into a new root {\bf \em r}$\,(T)$. That is the subtrees of $A$ and $B$ (if any) become the subtrees of
{\bf \em r}$\,(T)$.
We have $A\,+$ {\bf 1} = {\bf 1} $+\,A=A$. 

\item{{\bf Multiplication}. $T=A\cdot B$ is built by merging {\bf \em r}$\,(B)$ with each vertex $x\in A$ so that all the subtrees of {\bf \em r}$\,(B)$ become new subtrees of $x$.
We have $A\,\cdot$ {\bf 1} = {\bf 1} $\cdot \,A = A$.} 

\end{itemize}

In both operations it is immaterial in which order the subtrees are attached to the new parents. 
We also define  the operation stretch (symbol over-bar)  whose interest will be made clear in the following:

\begin{itemize}

\item{\bf Stretch}. $T=\bar{A}$ consists of a new root {\bf \em r}$\,(T)$ with $A$ attached as a subtree. 


\end{itemize}

In the notation stretch has precedence over multiplication, and multiplication has precedence over addition. Two propositions immediately follow: 

\begin{figure}
\label{fig2}
\begin{center}
\includegraphics[scale=0.6]{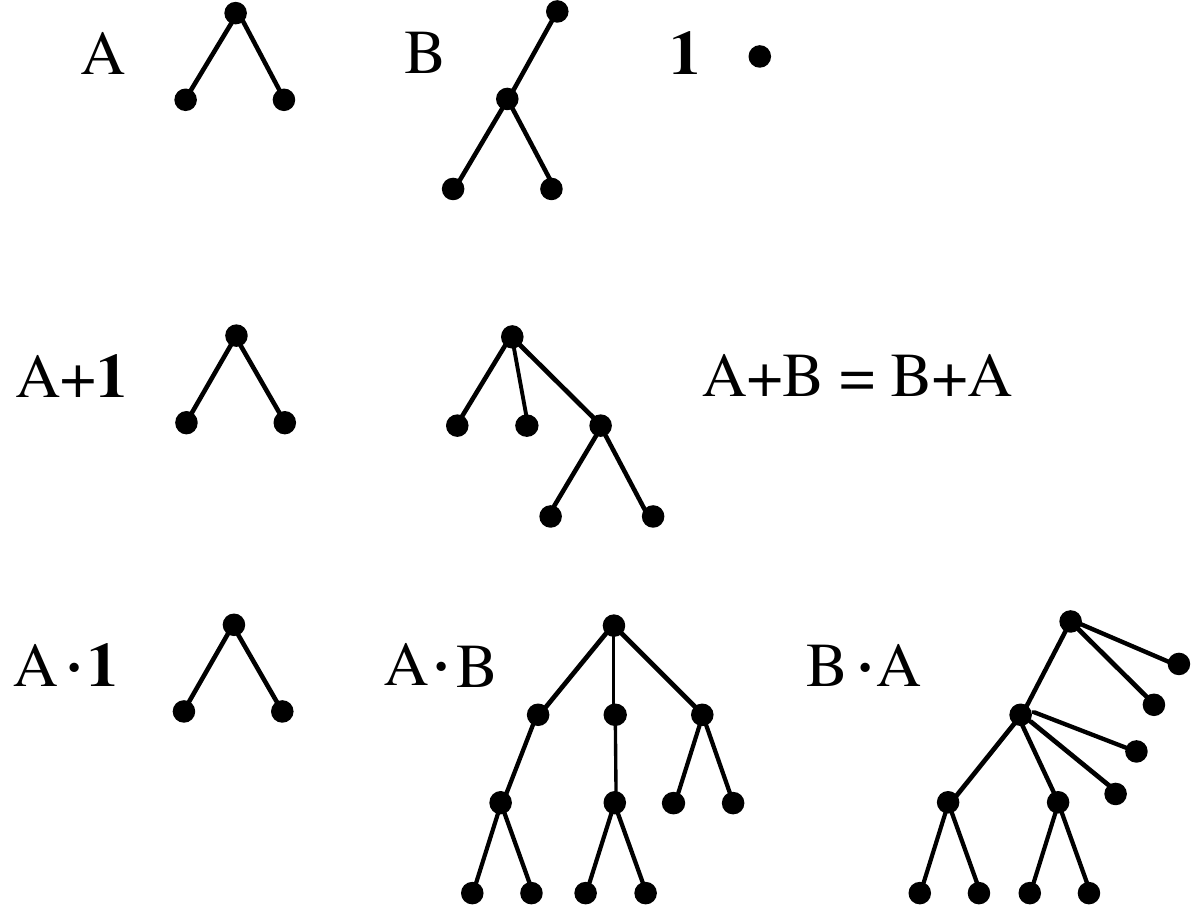}
\end{center}

\caption {Examples of addition and multiplication.}
\end{figure}

\begin{prop}\label{prop-numbers}
For $T=A+B$ we have $n_T= n_A+n_B-1$.  
For $T=A\cdot B$ we have $n_T=n_An_B$.  
For $T=\bar A$ we have $n_T=n_A+1$.  
\end {prop}

\begin{prop}\label{propad}
Addition is commutative and associative. That is $A+B=B+A$  and $(A+B)+C=A+(B+C)$.  
\end{prop}

For a positive integer $k>1$ and a tree $A$ we can define the product $T=kA$ (not to be confused with the product of trees) as the sum of $k$ copies of $A$. Due to Propositions \ref{propad} 
and \ref{prop-numbers}, 
the $k$ copies of $A$ can be combined in any order and we have $n_T=k\,n_A-k+1$. However, for any given $k$, the different trees of $n_T$  vertices obtained as a product $T=kA$ are only $f_{n_A}$, that is they constitute an exponentially small fraction of all the trees in ${\cal F}_{n_T}$. For example the ``even'' trees (obtained for $k$ even) are a small minority among all the trees with the same number of vertices.
Similarly we can define the stretch-product $U=k\bar A$ as $A$ stretched $k$ times, and we have $n_U=\,n_A+k$. Again for any given $k$, the trees of $n_U$  vertices obtained as a stretch-product $k\bar A$ are only $f_{n_A}$ and constitute an exponentially small fraction of all the trees in ${\cal F}_{n_U}$. 

For tree multiplication, associativity is simple but commutativity is more complicated. From the definition of multiplication we have with simple reasoning: 

\begin{prop}\label{propmas}
Multiplication is associative. 
\end {prop}

That is
$(A\cdot B)\cdot C = A\cdot (B\cdot C)$.
For a positive integer $k>1$ and a tree $A$ we can define the power $T=A^k$ as the product of $k$ copies of $A$. Due to Propositions \ref{propmas} and \ref {prop-numbers} the multiplications can be done in any order and we have $n_T=n_A^k$. Again, for any given $k$, the different trees of $n_T$  vertices obtained as $T=A^k$ are only $f_{n_A}$.

Multiplication is generally not commutative. For a product  $A\cdot B$ we consider the cases $n_A=n_B$ and $n_A>n_B$ (the case $n_A<n_B$ is symmetric), and pose the conditions below. Recall that, for any tree $X$, $e_X$ and $s_X$ respectively denote the number of leaves of $X$ and the number of subtrees of {\bf \em r}$\,(X)$.
For $n_A>n_B$ our 
conditions are only necessary. 

\begin{prop}\label{propmcom1}
For $n_A=n_B$ 
we have $A\cdot B = B\cdot A$ if and only if $A=B$.
\end {prop}

\begin{proof}
The if part is immediate. For the only if part let $T=A\cdot B$ and $U=B\cdot A$. From the construction of the two products
we immediately have $e_T=n_Ae_B$ and $e_U=n_Be_A$. If $T=U$ we have $e_T= e_U$ then $n_Ae_B=n_Be_A$, then $e_A=e_B$ since $n_A=n_B$. 
Note that $T$ and $U$ contain $e_A=e_B$ subtrees rooted in the former leaves of $A$ and $B$ respectively, each coinciding with $B$ and $A$ respectively. Each of these subtrees contains $n_B=n_A$ vertices, while all the other subtrees of $T,U$ contain a different number of vertices. Then for having $T=U$ the former two groups  of subtrees should be identical, that is each subtree coinciding with $B$ in $T$ must be equal to a subtree coinciding with $A$ in $U$. That is $A=B$.
\end{proof}


\begin{prop}\label{propmcom2}
For $n_A>n_B$ 
we have $A\cdot B = B\cdot A$ only if the following conditions are all verified:

\noindent (i) $n_a/e_A = n_B/e_B$;

\noindent (ii) $B$ is a proper subtree of $A$;

\noindent (iii) if $s_A\geq s_B$ all the subtrees of {\bf \em r}$\,(B)$ must be equal to some subtrees of {\bf \em r}$\,(A)$.
\end {prop}

\begin{proof}
Let $T=A\cdot B$ and $U=B\cdot A$.

\noindent {\em Condition (i)}. 
Immediate from the observation that $T=U$ implies $e_T=e_U$ (see the proof of Proposition 5).

\noindent {\em Condition (ii)}. 
As in the proof of Proposition 5, consider the subtrees of $T,U$ respectively attached to the former leaves of $A$ in $T$ and of $B$ in $U$. Since $n_Ae_B=n_Be_A$ (see the proof above) and $n_A>n_B$ we have $e_A>e_B$. In $T$ there are $e_A$ such subtrees of $n_B$ vertices and in $U$ there are $e_B$ such subtrees of $n_A$ vertices. For having $T=U$ the above subtrees of $T$ (all coinciding with $B$) should be present also in $U$ where, by the construction of $B\cdot A$, they must appear as subtrees of the copies of $A$ in $U$.

\noindent {\em Condition (iii)}. By construction the $s_B$ subtrees of {\bf \em r}$\,(B)$ appear also in $T$ as subtrees of {\bf \em r}$\,(T)$ where they are the ones with fewer vertices because all the others have at least $n_B$ vertices. 
And the $s_A$ subtrees of {\bf \em r}$\,(A)$ appear also in $U$ as subtrees of {\bf \em r}$\,(U)$ where they are the ones with fewer vertices because all the others have at least $n_A$ vertices. Note that all these other subtrees of {\bf \em r}$\,(U)$ have more vertices than the subtrees of {\bf \em r}$\,(B)$ since $n_A>n_B$.
For having $T=U$ the $s_B$ subtrees of {\bf \em r}$\,(B)$ that appear as subtrees of {\bf \em r}$\,(T)$ must be equal to $s_B$ subtrees of {\bf \em r}$\,(U)$ and, for what just seen about these subtrees, they must be equal to $s_B$ subtrees among the ones with fewer vertices, i.e. with subtrees of {\bf \em r}$\,(A)$. This also implies that if $s_A=s_B$ then $A=B$.
\end{proof}


The trees $3=A$ and $2=B$ of Figure 2 do not comply with conditions (i) and (ii) of Proposition~\ref{propmcom2} and
we have $A\cdot B=22$ different from $B\cdot A=20$. Commutative products are in fact quite rare. An example with $A\cdot B= B\cdot A$ is shown in Figure 5 where the three conditions of Proposition \ref{propmcom2} are verified. 
In this particular case we have $A=B^2$ hence $A\cdot B=B^3$.
Finally multiplication is generally not distributive over addition. From Proposition \ref{prop-numbers} we can immediately prove: 

\begin{figure}
\label{fig3}
\begin{center}
\includegraphics[scale=0.6]{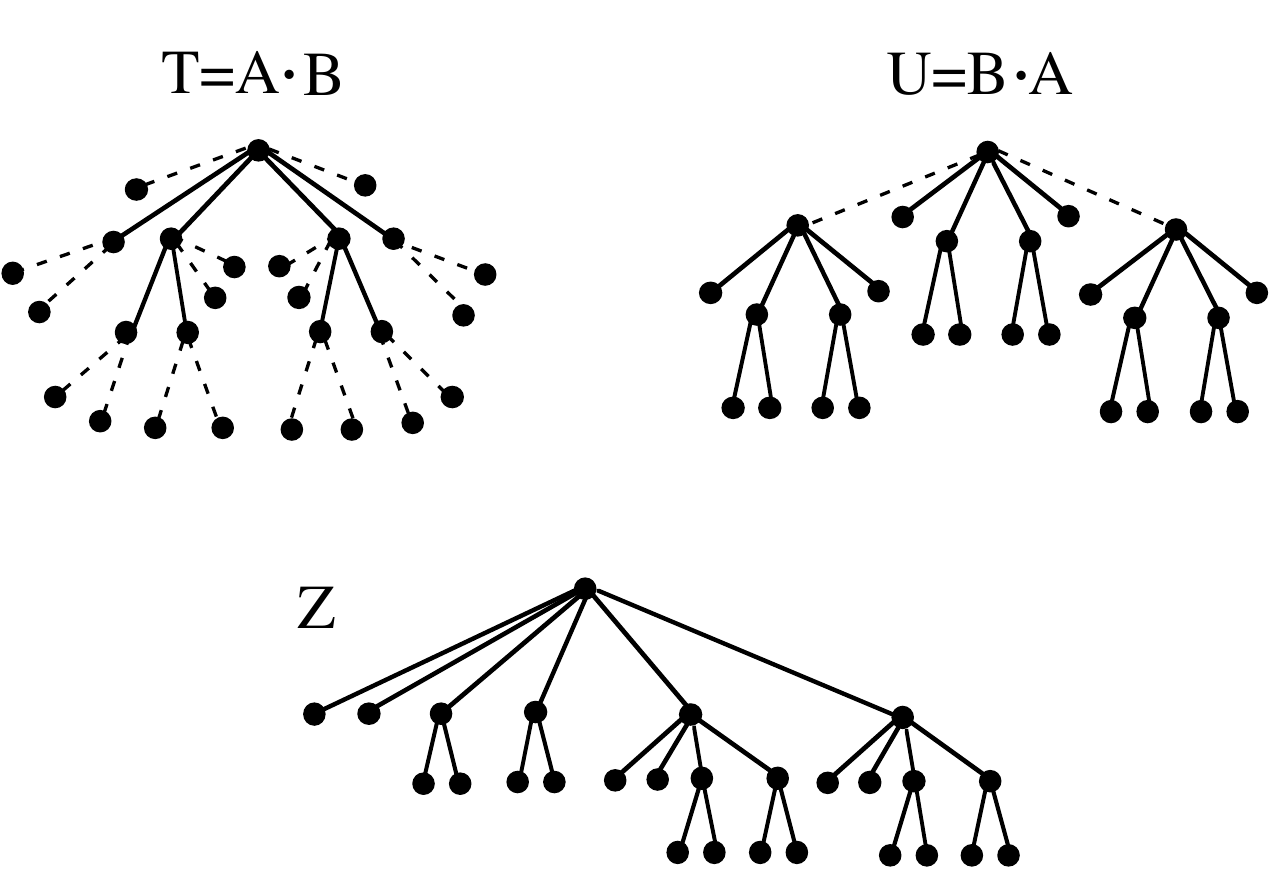}
\end{center} 

\caption {An example of commutative product $A\cdot B= B\cdot A$ for $B$ subtree of $A$. The two trees are shown in solid and dashed lines, respectively. $Z$ is $A\cdot B$ in canonical form. }
\end{figure}

\begin{prop}\label{propmdist}
$(A+B)\cdot C = A\cdot C + B\cdot C$ if and only if $C=\,${\bf 1}.  
\end{prop}

A basic fact about our arithmetic is that all trees can be generated by the single {\em generator} {\bf 1} using addition and stretch.\footnote{Stretch been included in the operation set to allow the construction of all trees starting from a finite set of generators. The reader may check that addition and multiplication, or stretch and multiplication, are not sufficient for this purpose. } Namely: 

\begin{itemize}

\item Tree {\bf 1} is the generator of itself.

\item Assuming inductively that each of the trees in ${\cal F}_i$ with $1\leq i\leq n-1$ can be generated by the trees of the preceding families, then each tree $T$ in ${\cal F}_n$ can also be generated. In fact if  
{\bf \em r}$\,(T)$ has one subtree $T_1$ then $T$ can be generated as $\bar T_1$; if {\bf \em r}$\,(T)$ has $k\geq2$ subtrees $T_1,T_2,\dots ,T_k$ then $T$ can be generated as $U+V$ where $U$ is $T$ deprived of $T_k$ and $V$ is $T$ deprived of $T_1,T_2,\dots ,T_{k-1}$.  

\end{itemize}


\section{Prime trees} \label{prime}

In the arithmetic of natural numbers the basic operations are addition and multiplication, with $x+0=x$ and $x\cdot 1=x$. Prime numbers under addition have no sense, since all $x$ greater than 1 can be constructed as the sum of two smaller terms other than 0 and $x$. In our arithmetic for trees, instead, primality occurs in relation with addition and multiplication.
In this whole section we refer to trees $T$ with $n_T>1$. We pose:

\begin{defin}\label{def-numbers}
{\em (i) $T$ is {\em prime under addition} {\em (}shortly {\em add-prime)} if can be generated by addition only if the terms are {\bf 1} and $T$ 
(tree {\bf 1} has a companion role of integer 0 in $I\! \! N$).}

\noindent {\em (ii) $T$ is {\em prime under multiplication} {\em (}shortly {\em mult-prime)} if can be generated by multiplication only if the factors are {\bf 1} and $T$.}
\end {defin}

The definition of mult-primality is the natural counterpart of the one of primality in $I\! \! N$. As it may be expected its consequences are not easy to study. 
For add-primality, instead, the situation is quite simple. We have:


\begin{prop}\label{prop-addprime}
$T$ is add-prime if and only if {\bf \em r}$\,(T)$ has only one subtree.
\end {prop}

\begin{proof}
By contradiction. {\em If part}: for an arbitrary tree $X=A+B$ with $A,B\neq$ {\bf 1}, {\bf \em r}$\,(X)$ has at least two subtrees, then $T\neq X$ for any pair $A,B\neq$ {\bf 1}.
{\em Only if part}: if {\bf \em r}$\,(T)$ has $k>1$ subtrees $T_1,\dots,T_k\neq$ {\bf 1} then $T=U+V$, where $U$ is equal to $T$ deprived of $T_k$ and $V$ is equal to $T$ deprived of $T_1,\dots,T_{k-1}$, with $U,V\neq$ {\bf 1}.
\end{proof}

As a consequence of Proposition \ref{prop-addprime} deciding if a tree is add-prime is computationally ``easy". From Proposition \ref{prop-addprime}, and from the construction given in the {\bf DR} rule we have:

\begin{prop} \label{num-addprime}
For $n\geq 2$ the number of add-prime trees is $f_{n-1}$.
\end{prop}


From Equation (1) we have: $f_{n-1}/f_n \rightarrow \,\sim 0.34$ for $n \rightarrow \infty$, that is the add-prime trees in ${\cal F}_n$ are asymptotically about one third of the total. Each of the remaining {\em add-composite} (i.e., non add-prime) trees $T$ can be uniquely factorized in $s_T$ factors.

For mult-primality we start with two immediate statements respectively derived from Proposition \ref{prop-numbers}, and from the definition of multiplication for trees with at least two vertices. More complex conditions for mult-primality can be found in~\cite{L15}.

\begin{prop}\label{prop-mult1}
If $n$ is a prime number all the trees with $n$ vertices are mult-prime.
\end {prop}

\begin{prop}\label{prop-mult3}
If {\bf \em r}$\,(T)$ has only one subtree then $T$ is mult-prime.
\end {prop}

The converse of Propositions \ref{prop-mult1} and \ref{prop-mult3} do not hold in our arithmetic. That is if $n_T$ is a composite number or {\bf \em r}$\,(T)$ has more than one subtree, tree $T$ may still be mult-prime. In a sense mult-prime trees are more numerous than primes in $I\! \! N$. For example out of the twenty trees in ${\cal F}_6$ (see Figure 2) only trees 20, 22, 24, and 28 are {\em mult-composite} (i.e. non mult-prime), as they can be built as $2\cdot 3$, $3\cdot 2$, $4\cdot 2$, and $2\cdot 4$, respectively. 


Since if $n_T$ is prime $T$ is mult-prime, and the problem of deciding if $n_T$ is prime is polynomial in $\log n_T$, 
deciding if $T$ is mult-prime is straightforward for $n_T$ prime. However the problem is difficult  for $n_T$ composite because $T$ may be mult-prime or mult-composite. 
An algorithm for $n_T$ composite may consist of building all the products $A\cdot B$ and $B\cdot A$ of two trees $A, B$ of $a,b$ vertices respectively for all the factorizations of $n_T$ as $a\cdot b$, and comparing $T$ with these products looking for a match. However this method is impracticable unless $n_T$ is very smal, then we must find a different way to decide  mult-primality. To this end consider a property of product trees 
based on the observation that, if
$T=A\cdot B$, all the subtrees of {\bf \em r}$\,(B)$ are also subtrees of {\bf \em r}$\,(T)$. Namely:

\begin{prop}\label{prop-B}
Let $T=A\cdot B$ with $A,B \neq$ {\bf 1}, and let $Y$ be a subtree of {\bf \em r}$\,(B)$ with maximum number $n_Y$ of vertices. Then the subtrees of {\bf \em r}$\,(B)$ are exactly the subtrees of {\bf \em r}$\,(T)$ with at most $n_Y$ vertices. 
\end{prop}

\begin{proof}
Since $T=A\cdot B$, the subtree $Y$ has been inserted at {\bf \em r}$\,(T)$ as the largest subtree of {\bf \em r}$\,(B)$. Then also
the subtrees of {\bf \em r}$\,(T)$ with at most $n_Y$ vertices must have been inserted at {\bf \em r}$\,(T)$ as subtrees of 
{\bf \em r}$\,(B)$ since they have too few vertices for deriving from former subtrees of {\bf \em r}$\,(A)$ whose vertices are merged with $B$ in $T$. Furthermore the remaining subtrees of {\bf \em r}$\,(T)$ cannot be subtrees of
{\bf \em r}$\,(B)$ since they have too many vertices by the hypothesis that $Y$ is a largest subtree of  {\bf \em r}$\,(B)$.
\end{proof}

In the mult-composite tree $Z$ of Figure 5, if the first subtree of {\bf \em r}$\,(Z)$ (containing one vertex) is a subtree of maximal cardinallity of one of the factors, $B$ in this case, then $B$ consists of a root plus the first two subtrees of {\bf \em r}$\,(Z)$. 
Similarly, if the third subtree of {\bf \em r}$\,(Z)$ is a subtree of maximal cardinality of one of the factors, $A$ in this case, then $A$ consists of a root plus the first four subtrees of {\bf \em r}$\,(Z)$. We pose:

\begin{notat} \label{notat1}
For an arbitrary tree $T$:
{\em (i)} $G_1,\dots,G_r$ are the groups of subtrees of {\bf \em r}$\,(T)$ with the same number $g_1,\dots, g_r$ of vertices, $g_1<g_2<\dots<g_r$;
$\;\;${\em (ii)} $H_i=\bigcup^i_{j=1}G_j$, $1\leq i\leq r$, i.e. each $H_i$ is the group of subtrees of {\bf \em r}$\,(T)$ with up to $g_i$ vertices. 
\end{notat}
Based on Propositions \ref{prop-B} and Notation \ref{notat1} we can build the primality Algorithm MP of Figure 6 that requires polynomial time in the number of vertices. Since all trees with a prime number $n$ of vertices are mult-prime, MP is intended for testing trees with $n$ composite. However MP works for all trees and can always be applied to avoid a preliminary test for the primality of $n$.


\begin{figure} \label{MP}

\begin{center}
\fbox{
\begin{minipage}{12.8cm}
{\bf algorithm} MP($T$,$\,n$)

\vspace{2mm}

{\bf 1.} CF($T$,$\,n$);

\vspace{1mm}
\hspace{5mm} // transform $T$ in canonical form with Algorithm CF of Figure 4

\vspace{1mm}
{\bf 2.} {\bf let} $H_1,\dots,H_r$ be the groups of subtrees of {\bf \em r}$\,(T)$ as in Notation \ref{notat1};

\vspace{1mm}
{\bf 3.} {\bf for} $1\leq i\leq r-1$

\vspace{1mm}
\hspace{10mm} {\bf copy} $T$ into $Z$;

\vspace{1mm}
\hspace{10mm} {\bf traverse} $Z$ in preorder

\vspace{1mm}
\hspace{17mm} {\bf forany} vertex $x$ encountered in the traversal

\vspace{1mm}
\hspace{24mm} {\bf if} $x$ has all the subtrees of $H_i$ {\bf erase} these subtrees in $Z$

\vspace{1mm}
\hspace{24mm} {\bf else} {\bf exit} from the $i$-th cycle;

\vspace{1mm}
\hspace{17mm} {\bf return} MULT-COMPOSITE;
 
\vspace{1mm}
{\bf 4.} {\bf return} MULT-PRIME.

\end{minipage}
}

\end{center}

\caption {\small Structure of Algorithm MP for deciding if a tree $T$ of $n$ vertices is mult-prime. 
} 

\end{figure}


\begin{prop}\label{prop-MP}
Mult-primality of a tree $T$ can be decided in time polynomial in $n_T$. 
\end {prop}

\begin{proof}
Refer to Algorithm MP.  {\em Correctness}. Only step {\bf 3} requires an analysis. $Z$ is the changing version of $T$ and is restored at each $i$-th cycle. If one of the groups $H_i$ of subtrees can be erased from $Z$ at all vertices encountered in the traversal, the cycle is completed and the algorithm terminates declaring that $T$ is mult-composite. In fact tree $B$, whose root  has the subtrees in $H_i$, is one of the factors of $T$ (see Proposition \ref{prop-B}). If none of the $i$-cycles can be completed, that is no $H_i$ can be found as being the group of subtrees of $x$ in all vertices $x$ of $Z$, the tree $T$ is mult-prime as declared in step {\bf 4}.

\noindent {\em Complexity}. A superficial analysis of the algorithm is the following. Step {\bf 1} requires O($n^2$) time as discussed for Algorithm CF. Step {\bf 2} is executed with a linear time scan because the tree is now in canonical form and the number of vertices in each subtree of the root has been computed by algorithm CF in step {\bf 1}. Step {\bf 3} requires O($n$) copy operations of $T$ into $Z$ in O($n^2$) time, and O($n$) traversals each composed of O($n$) steps, for a total of O($n^2$) steps. At each step at vertex $x$ the subtrees in $H_i$ must be compared with the subtrees of $x$ with the same cardinality; this can be done by representing such subtrees with their binary sequences $S$ and comparing these sequences. In the worst case vertex $x$ has O($n$) subtrees of length O($n$), so that building and comparing all the sequences takes time O($n^2$), and the total time required by step {\bf 3} is O($n^4$). Note that this analysis is very rough because the number of vertices of $T$ decreases during the traversal, so the stated bound O($n^4$) is exceedingly high.
\end{proof}

Note that if $T$ is mult-composite Algorithm MP allows to find a pair of factors $A,B$ at no extra cost, with $B$ mult-prime. In fact, if a cycle $i$ of step {\bf 3} is completed, the algorithm is interrupted on the {\bf return} statement and the group $H_i$ contains exactly the subtrees of {\bf \em r}$\,(B)$, while the tree $Z$ is reduced to $A$. 
In particular $B$ is the last factor of a product of mult-prime trees, with $T=T_1\cdot T_2 \dots \cdot T_k \cdot B$. If Algorithm MP is not interrupted with the {\bf return} statement, all these factors can be detected. 
As a consequence we have:

\begin{prop}\label{prop-unique}
Mult-factorization of any tree $T$ is unique.
\end{prop}

\begin{proof}
By contradiction assume that $T$ has two different factorizations $T_1\cdot T_2 \dots \cdot T_k $ and $S_1\cdot S_2 \dots \cdot S_h$ in multi-prime factors. Tracing back from $k$ and $h$, let $T_i$ and $S_j$ be the first pair of factors encountered with $T_i\neq S_j$. Then we have $T_1\cdot T_2 \dots \cdot T_i= S_1\cdot S_2 \dots \cdot S_j$.
By Proposition~\ref{prop-B} $T_i$ must contain $S_j$ as one of its factors (or vice-versa), against the hypothesis that $T_i$ is mult-prime. 
\end{proof}


Finally note that counting the number of add-prime trees is simple (Proposition \ref{num-addprime}), but an even approximate count for mult-prime trees is much more difficult. 
We pose:

\vspace{2mm}
\noindent{\bf Open problem}. {\em For a composite integer $n$ determine the number of mult-prime trees of $n$ vertices.}


\section{Negative trees and tree equations}

Once addition and multiplication are known, it is natural to define the inverse operations.

We define the {\bf subtraction} $A=T-B$ if and only if all the subtrees of {\bf \em r}$\,(B)$ are also subtrees of {\bf \em r}$\,(T)$. Then $A$ equals $T$ deprived of such subtrees. This is the inverse of the addition $T=A+B$. We have $T\,-$ {\bf 1} = $T$.

We define the {\bf division} $A=T/B$ if and only if there exists a subset $\Psi$ of the vertices of $T$ such that each $v\in \Psi$ has exactly the subtrees of {\bf \em r}$\,(B)$, and the tree $T'$ obtained as $T$ deprived of such subtrees has exactly the vertices of $\Psi$. Then $A=T'$.
This is the inverse of the multiplication $T=A\cdot B$. We have $T\,/$ {\bf 1} = $T$.

Also the operation of stretch has an inverse. We define the {\bf un-stretch} 
(symbol underline) if and only if {\bf \em r}$\,(A)$ has exactly one subtree $T$, and we pose 
$\underline{A}=T$. 
In the notation un-stretch has precedence over multiplication and stretch has precedence over un-stretch.

\vspace{1mm}
As negative numbers arose from subtraction in integer arithmetic, the more intriguing concept of negative trees arises here from tree subtractions. We propose the following definition. All the vertices of a tree $T$ are either {\bf positive} (then $T$ is positive) or {\bf negative} (then $T$ is negative), except for the root that is {\bf neutral}. Positive and negative vertices are respectively indicated with a black dot or an empty circlet. The root is also indicated with a black dot. Changing the sign of a tree amounts to changing the nature of all its vertices except for the root. Tree {\bf 1} is neutral and we have {\bf 1} $= -$ {\bf 1}.

Addition and subtraction between $A$ and $B$ keep their definition with the additional condition that if $A$ is positive and $B$ is negative all the subtrees of {\bf \em r}$\,(B)$ are also subtrees of {\bf \em r}$\,(A)$ or vice-versa, and positive and negative subtrees with identical shape cancel each other out in the result (See Figure 7). Multiplication and division between $A$ and $B$ also keep their definition with the additional condition that if $A$ and $B$ are both positive or both negative the result is positive, otherwise is negative. 

\begin{figure}
\label{fig-neg}
\begin{center}
\includegraphics[scale=0.8]{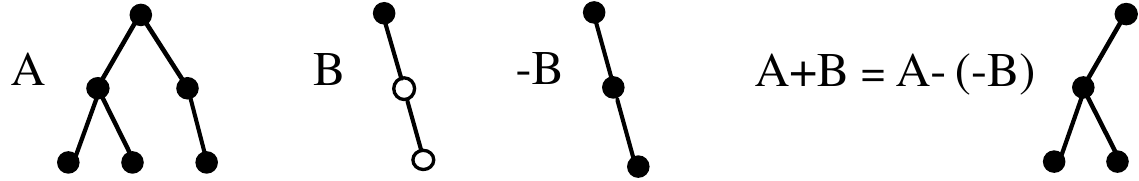}
\end{center} 

\caption {Addition between a positive tree $A$ and a negative tree $B$. }
\end{figure}

At this point we may open a window on {\bf tree equations} whose terms have all the nature of a tree, but integers may appear as multiplicative coefficients or exponents. In a sense they are companions of the Diophantine equations with integers, but the solutions are now required to be trees. We may consider equations of different degrees with different number of variables, ask questions on the existence and on the number of solutions, study the computational complexity of finding them. In fact we give only some examples, leaving the field essentially open.

Denote trees and integers with capital and lower case letters respectively. The simplest  equation is linear and has only one unknown $X$. We put:

\vspace{2mm}
\hspace{5mm}$aX+C=\,${\bf 1}, \hspace{2mm}i.e.  $\;aX=-C$ \hfill(2)

\vspace{2mm}
\noindent Equation (2) admits exactly one solution if and only if the $s_C$ subtrees of {\bf \em r}$\,(C)$ can be divided in $g\geq 1$ groups $G_1,\dots,G_g$ of identical subtrees, where each $G_i$ has cardinality $k_ia$ for $k_i\geq 1$, see example E1 in Figure 8. In this case $X$ has $s_C/a$ subtrees that can be divided in $g$ groups of $k_i$ subtrees identical to the ones of $G_i$. This solution can be easily built in time polynomial in $n_C$ starting with the transformation of $C$ in canonical form. Note that $X$ and $C$ have opposite sign.

\begin{figure}
\label{fig-eq1}
\begin{center}
\includegraphics[scale=0.8]{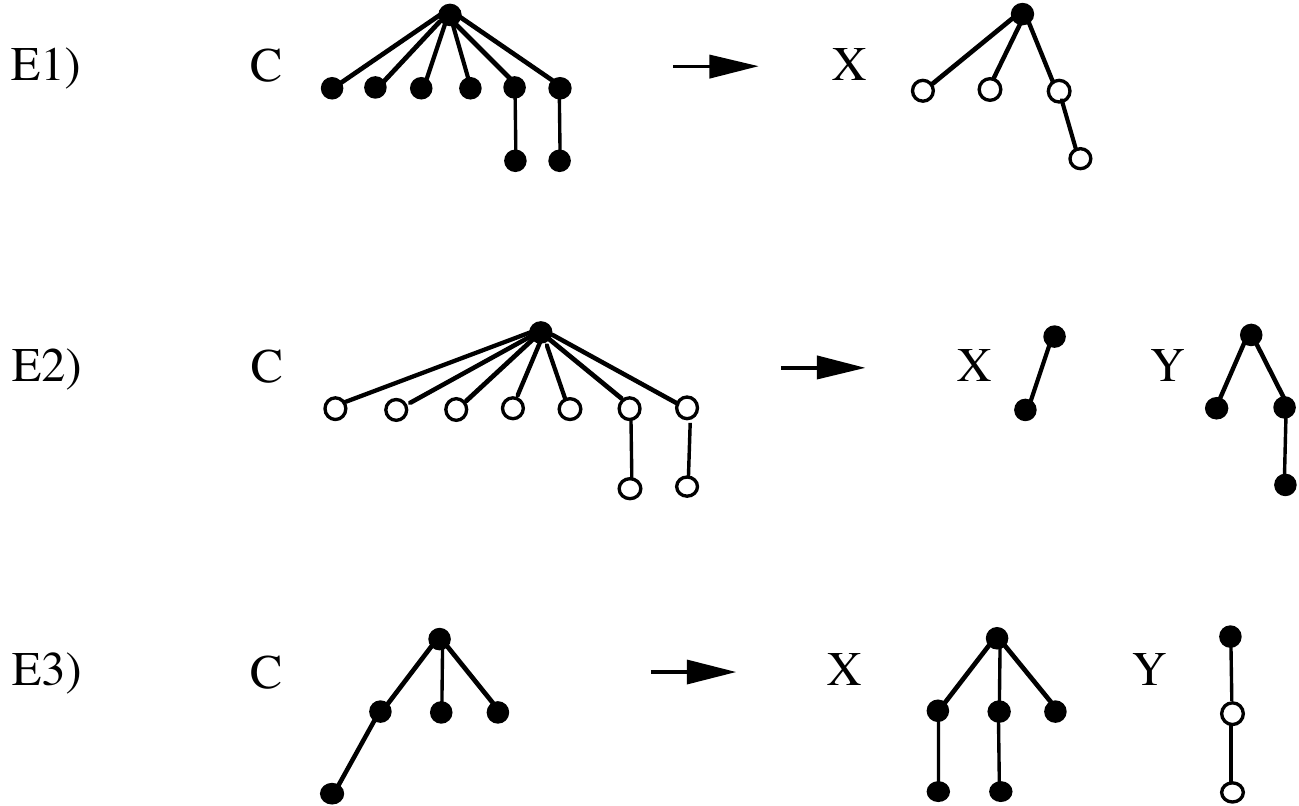}
\end{center} 

\caption {Solution of the tree equations: {\bf E1}: $2X+C\,$={\bf 1}. {\bf E2}: $3X+2Y+C=\,${\bf 1}. 
\newline 
{\bf E3}: $2X+3Y+C=\,${\bf 1}.}
\end{figure}

\vspace{1mm}
A standard linear tree equation in two unknowns $X,Y$ can be expressed as:

\vspace{2mm}
\hspace{5mm}$aX+b$\,$Y+C=\,${\bf 1}, \hspace{2mm}i.e  $\;aX+b$\,$Y=-C$ \hfill(3)

\vspace{2mm}
\noindent This equation is the companion of the diophantine equation $ax+by=c$ widely used in modular algebra, that admits an integer solution if and only if $c$ is divided by $gcd(a,b)$. So applying Proposition \ref{prop-numbers} to the trees of equation (3) we have $a\,n_X+b\,n_Y=n_C+a+b-1$ and a necessary condition for the existence of a tree solution is that $n_C+a+b-1$ divides $gcd(a,b)$, as in the examples E2, E3 of Figure 8 where $gcd(a,b)=1$. In general equation (3) admits a solution if and only if one of the two non trivial Conditions 1 and 2 below hold, corresponding respectively to trees $X,Y$ of equal sign or of opposite sign. In both cases the solution can be built in time polynomial in $n_C$. We have:

\vspace{2mm} 
\noindent {\bf Condition 1}. The $s_C$ subtrees of {\bf \em r}$\,(C)$ can be divided in $g\geq 1$ groups $G_1,\dots,G_g$ and $h\geq 1$ groups $H_1,\dots,H_h$ of identical subtrees, where each $G_i$ has cardinality $g_ia$ for $g_i\geq 1$ and each $H_i$ has cardinality $h_ib$ for $h_i\geq 1$. 
In this case $X$ has $\sum_{i=1}^gg_i$ subtrees divided in $g$ groups of $g_i$ subtrees identical to the ones of $G_i$; and $Y$ has $\sum_{i=1}^hh_i$ subtrees divided in $h$ groups of $h_i$ subtrees identical to the ones of $H_i$.
This solution can be built in time polynomial in $n_C$. Note that $X$ and $Y$ have the same sign, and $C$ has opposite sign. See Equation E2 in Figure 8.

\vspace{2mm}
\noindent  {\bf Condition 2}. Let the unknown trees $X$ and $Y$ have opposite sign. W.l.o.g. let the subtrees of {\bf \em r}$\,(X)$ be divided in $k+h$ groups $G_1,\dots, G_{k+h}$ of identical subtrees, and the subtrees of {\bf \em r}$\,(Y)$ be divided in $k$ groups $H_1,\dots, H_{k}$ of identical subtrees, with $k\geq 1$ and $h\geq 0$. 
And let the subtrees of {\bf \em r}$\,(C)$ be divided in $k+h$ groups $C_1,\dots, C_{k+h}$ of identical subtrees. 
$x_i$, $y_i$, $c_i$ respectively denote the cardinalities of $G_i,H_i,C_i$. 

To allow the addition $\;aX+b$\,$Y$ 
the subtrees in $H_i$ must be identical to the ones in $G_i$ for $1\leq i \leq k$; the subtrees in $C_i$ must be identical to the ones in $G_i$ for $1\leq i \leq k+h$; and we have the system of diophantine equations:

\vspace{1mm}
$a\,x_i-b\,y_i=c_i   \hspace{2mm}$ for $\;1\leq i\leq k$\hfill(i)

$a\,x_i=c_i   \hspace{12mm}$ for $\;k+1\leq i\leq k+h$\hfill(ii)

\vspace{1mm}
\noindent whose integer solutions (if any) state that the $a$ copies of the subtrees of $G_i$ suffice to elide the $b$ copies of the subtrees of $H_i$ in $C$, for $ i\leq k$; and $a$ copies of the subtrees in $G_i$ appear as subtrees of $C_i$, for $i> k$.  
The system can be solved under the conditions:

\vspace{1mm}
$c_i / GCD(a,b)$ integer  \hspace{2mm} for $\;1\leq i\leq k$\hfill(iii)

$c_i /a$ integer  \hspace{17mm} for $\;k+1\leq i\leq k+h$\hfill(iv)

\vspace{1mm}
\noindent for a value of $k$ established as the minimum value for which condition (iv) holds (this fixes also the value of $h$). Then if all conditions (iii) hold the system is solved in time polynomial in $n_C$ and two trees $X$, $Y$ satisfying equation (3) are immediately built from the values of $x_i$, $y_i$, out a potentially infinite number of solutions. In particular note that, for all $i$, the values $x_i$, $y_i$ must be both positive to represent subset cardinalities. If this does not happen, an alternative positive solution is built from the other by standard methods. See equation E3 in Figure 8.

\vspace{3mm}

Higher degree equations are more difficult to handle. For the quadratic tree equation:

\vspace{2mm}
\hspace{5mm}$aX^2+b$\,$Y+C=\,${\bf 1}, \hspace{2mm}i.e  $\;aX^2+b$\,$Y=-C$ \hfill(4)

\vspace{2mm}
\noindent a necessary condition for the solution is the existence of two integers $n_X, n_Y$ satisfying the algebraic equation $a\,n_X^2+b\,n_Y=n_C+a+b-1$, a well known NP-complete problem. To find a reasonably interesting approach for deciding whether equation (4) has a solution is left as an open problem.

\vspace{1mm}
A ``more ambitious'' problem can be expressed as:

\vspace{2mm}
\hspace{5mm}$X^n+Y^n=Z^n$\hfill(5)

\vspace{2mm}
\noindent with the question of deciding if equation (5) has a tree solution $X,Y,Z$ for any $n\geq 2$. In fact even for $n=2$ the problem is not simple. Due to Proposition \ref{prop-numbers} we have the necessary condition $n_X^2+n_Y^2-1=n_Z^2$ for its solution, i.e. the existence of a ``quasi-Pythagorean'' triple of integers. In fact such triples exist, as for example $\{4, 7, 8\}$, but the existence of Pythagorean trees with such numbers of vertices is left as an open problem.


\section{Possible applications and extensions}\label{applications}

While the major purpose of the present study is  the one of defining arithmetic concepts outside the realm of numbers, let us briefly discuss what the role of our proposal in applications might be.

Essentially all trees used in computer algorithms are rooted, and different families have been defined among them to deal with particular problems.  
We do not put any restriction on the tree structure. The trees considered here simply correspond to nested sets as for example hierarchical structures in computer science; or office plans in business organization; or phylogenetic trees in biology, etc. Note that the subtrees are essentially unordered at any vertex, although they must be stored in some standard form to be represented, e.g. following an alphanumeric label order of similar. Or, of course, in our canonical order.

Two main actions are generally required in a hierarchical structure. Namely: (i) add a new subtree $B$ to the root of a tree $T$; or (ii) join two independent trees $A,B$ to form a new tree $T$ with $A,B$ subtrees of the root. In our arithmetic, action (i) is represented as $T=T+\bar B$; and action (ii) is represented as $T=\bar A +\bar B$.
Both actions can be respectively undone as: $T=T- \bar B$; and $A=\underline{T-\bar B}$, $B=\underline{T-\bar A}$.

An important extension of action (i) is inserting a new subtree $A$ at a given vertex $v$ of $T$. This is obtained by an iterative operation along the path $\pi = (v_0,v_1,\dots,v_k)$, from {\bf \em r}$\,(T)=v_0$ to $v=v_k$. Letting $T_0, \dots,T_k$ be the subtrees rooted at vertices $v_0,\dots,v_k$, hence $T=T_0$, we set $S_i=T_i-\bar T_{i+1}$ for $i=0,1,\dots, k-1$; then we set $T_k=T_k+\bar A$; then we set $T_{i-1}=S_{i-1}+\bar T_i$ for $i=k,k-1,\dots,1$, where $T_0=T$ gives the transformed tree. A similar operation is required to extract a subtree $A$ at vertex $v$.
Propositions \ref{propad} and \ref{propmas} hold for the subtrees rooted at $v$, with obvious effects on the whole tree.

Other operations can be considered and their representation investigated along the lines above. In particular multiplication may be performed on subtrees only, and even be limited at leaves.  
Note that, even though multiplication could find fewer applications than addition and stretch, it may be useful in data compression because the information contained in a product $A\cdot B$ is fully present in its factors, thereby reducing the storage space needed for the product from $\Theta(n_A\cdot n_B)$ to $\Theta(n_A+n_B)$.
So the concept of primality may be of practical interest in the reverse-engineering operation of deciding if a tree has been generated as a product.


\section{Other studies on tree arithmetic}

Up to now only one major line of research, that we call LBY, has been directed to defining arithmetic on trees. Opened by J.L. Loday {\em et al} in connection with dendriform algebras \cite{L+01}, it was then developed by J.L. Loday himself who gave a full description of arithmetic operations on binary trees and their properties, showing an embedding of $I\! \! N$ in the subsets of all binary trees of $n$ vertices \cite{L02}. A. Bruno and D. Yasaki worked on Loday's theory introducing primality and counting properties on subsets of trees in \cite{BY11}. 
LBY is limited to binary trees, which carries simpler consequences than in our general case. A non-commutative  tree addition is defined in LBY, attaching the second addend to a deepest leaf of the first one, and this operation is given in two versions to express any tree by addition from one generator (as in our proposal two different operations are needed). From this construction stems a definition of tree multiplication to produce trees different from our products. Several interesting properties are derived, including some counting arguments on the different families of trees built. The most relevant extension done by Bruno and Yasaki over Loday's theory is the definition and treatment of prime trees under multiplication. 
Aside from proceeding with similar purposes, none of the definitions and results of LBY applies to our theory, or vice-versa.

Another study on tree arithmetic, due to R. Sainudiin, is aimed at using binary trees for treating mapped partitions of a special class of intervals \cite{S14}, and has nothing to share with LBY and with our theory. 
None of these works deals with aspect of computational complexity related to the operations on trees.

Along an independent line of research several papers  have been directed to define graph multiplication, from the seminal work of G. Sibidussi~\cite{S60} to the one of B. Zmazek and J. Zerownik~\cite{ZZ07}. In this context prime graphs and graph factorization have been considered under various operations of multiplication, see~\cite{BW}. Again, if applied to trees as special graphs, all the definitions and results on tree multiplication are unrelated to ours.

We finally note that a preliminary work with partial overlapping with the present paper was deposited as an earlier {\em arXiv} manuscript~\cite{L15}. Such a version did not include negative trees and tree equations.




\newpage
\appendix
\section*{Appendix}
\label{app}

\vspace{5mm}
\begin{center}
\begin{tabular}
{||l|l||l|l||}
\hline

1  & 10 &                    &  \\\hline
 &  &                           18 & 110101010100 \\\hline
2  & 1100 &               19 & 110101011000  \\\hline
  &  &                           20 & 110101101000  \\\hline
3  & 110100 &           21 & 110101110000  \\\hline
4  & 111000 &           22 & 110110011000  \\\hline
  &  &                           23 & 110110101000  \\\hline
5  & 11010100 &      24 & 110110110000  \\\hline
6  & 11011000 &      25 & 110111010000 \\\hline
7  & 11101000 &      26 &  110111100000  \\\hline
8  & 11110000 &      27 &  111001101000 \\\hline
  &  &                           28 &  111001110000  \\\hline
9  & 1101010100 & 29 &  111010101000 \\\hline
10 &1101011000 & 30 &  111010110000 \\\hline
11 &1101101000 & 31 &  111011010000 \\\hline
12 &1101110000 & 32 &  111011100000 \\\hline
13 &1110011000 & 33 &  111100110000 \\\hline
14 &1110101000 & 34 &  111101010000 \\\hline
15 &1110110000 & 35 &  111101100000 \\\hline
16 &1111010000 & 36 &  111110100000 \\\hline
17 & 1111100000 &  37 &  111111000000  \\\hline
\end{tabular} 
\end{center}

\vspace{1mm}
\large
\begin{center}
The binary sequences representing the trees of the first six canonical families.
\end{center}


\small\begin{thebibliography}{10}\setlength{\itemsep}{-2mm}



\bibitem{BW}  N.~Bray and E.W.~Weisstein. \newblock Graph Product. {\em Math World} - A Wolfram Web Resource. \newblock {\em http://mathworld.wolfram.com/GraphProduct.html }

\bibitem{BY11} A.~Bruno and D.~Yasaki.\newblock The Arithmetic of Trees. \newblock {\em Involve} 4 (1) (2011) 1-11.

\bibitem{F03} S.~Finch.\newblock Two Asymptotic Series. 
\newblock {\em www.people.fas.harvard.edu/~sfinch/}


\bibitem{L+01} J.L.~Loday, A.~Frabetti, F.~Chapoton, and F.~Goichot. \newblock Dialgebras and related operands. \newblock {\em  Lecture Notes in Mathematics} 1763, Springer-Verlag, Berlin (2001).

\bibitem{L02}  J.L.~Loday. \newblock Arithmetree. {\em J. Algebra} 258 (1) (2002) 275Ð309.

\bibitem{L15}  F.~Luccio. \newblock Arithmetic for Rooted Trees. {\em arXiv:1510.05512v2}. (2015).

\bibitem{PR94}  J.M.~Plitkin and J.W.~Rosenthal. \newblock How to obtain an asymptotic expansion of a sequence from an analytic identity satisfied by its generating function. {\em J. Australian Math Soc. Ser.} A56 (1994) 131-143.

\bibitem{S60}  G.~Sabidussi. \newblock Graph Multiplication. {\em Math. Z.} 72 (1960) 446-457.

\bibitem{S14} R.~Sainudiin.\newblock Algebra and Arithmetic of Plane Binary Trees: Theory $\&$ Applications of Mapped Regular Pavings.\newline
\newblock {\em www.math.canterbury.ac.nz/r.sainudiin/talks/MRP$\_$UCPrimer2014.pdf }

\bibitem{Z80}  S.~Zaks. \newblock Lexicographic Generation of Ordered Trees. {\em Theoretical Computer Science} 10 (1980) 63-82.

\bibitem{ZZ07}  B.~Zmazek and J.~Zerownik. \newblock Weak Reconstruction of Small Product Graphs. {\em Discrete Mathematics} 307 (2007) 641-649.
\end{thebibliography}
\end{document}